\documentclass[a4paper,11pt]{article}
\usepackage{amsmath,amsthm,amssymb,amsfonts,latexsym}
\usepackage{graphicx}
\usepackage{epsfig}
\usepackage{framed}
\usepackage{boxedminipage}
\usepackage{enumitem}
\usepackage{xspace}
\usepackage[T1]{fontenc}
\usepackage{microtype}
\usepackage{scalefnt}
\usepackage[margin=1in]{geometry}

\usepackage{enumerate}
\usepackage{amsmath,amsfonts,amssymb,amsthm,latexsym}
\usepackage{mathrsfs}
\usepackage{xspace}

\usepackage{framed}
\usepackage{graphicx}
\usepackage{epsfig}

\usepackage{hyperref}
\hypersetup{
    colorlinks,
    linkcolor={blue},
    citecolor={blue},
    urlcolor={blue}
}

\newtheorem{theorem}{Theorem}
 
\newtheorem{definition}[theorem]{Definition}
\newtheorem{proposition}[theorem]{Proposition}
\newtheorem{lemma}{Lemma}[section]

\newtheorem{observation}[lemma]{Observation}

\newtheorem{reduction}[lemma]{Reduction Rule}

\makeatletter
\@addtoreset{claim}{lemma}
\@addtoreset{claim}{theorem}
\@addtoreset{claim}{observation}
\@addtoreset{claim}{proposition}
\@addtoreset{claim}{corollary}
\@addtoreset{subclaim}{claim}
\makeatother

\usepackage{thm-restate}

\newcommand{\NP}{\textsf{NP}\xspace}

\newcommand{\NPH}{\textsf{NP}-hard\xspace}
\newcommand{\NPC}{\textsf{NP}-complete\xspace}

\newcommand{\Yes}{\normalfont\textsf{Yes}}
\newcommand{\No}{\normalfont\textsf{No}}

\renewcommand{\O}{{\mathcal O}}

\newcommand{\paradefn}[4]{
  \vspace{2mm}
  \noindent\fbox{
  \begin{minipage}{0.97\textwidth}
  \vspace{1mm}
  \begin{tabular*}{1.01\textwidth}{@{\extracolsep{\fill}}lr} \textsc{\underline{#1}}  & {\textsf{Parameter:}} #3\\ \end{tabular*}
  {\textsf{Input:}} #2\\
  {\textsf{Question:}} #4
  \vspace{1.5mm}
  \end{minipage}
  }
}

\newcommand{\probdefn}[3]{
    \noindent\fbox{
        \begin{minipage}{.97\textwidth}
            \vspace{2mm}
            {{\textsc{\underline{#1}}}}  \\ 
            {\textsf{Input:}} #2  \\
            {\textsf{Output:}} #3
            \vspace{1.5mm}
        \end{minipage}
    }
}

\newcommand{\myurl}[2]{\url{#1}\xspace}
\newcommand{\useless}[1]{}
\usepackage{afterpage}

\newcommand{\calC}{\mathcal{C}\xspace}

\newcommand{\calF}{\mathcal{F}\xspace}

\newcommand{\calP}{\mathcal{P}\xspace}

\begin{document}
\begin{titlepage}
    
    \title{On Fault Tolerant Feedback Vertex Set} 
     
    \author{
        Pranabendu Misra\thanks{Max Planck Institute for Informatics, Saarbrucken, Germany. \texttt{pmisra@mpi-inf.mpg.de}}
    }
    \date{}
    \maketitle
    
    \thispagestyle{empty}
    
    \begin{abstract}
        The study of fault-tolerant data structures for various network design problems is a prominent area of research in computer science.
        Likewise, the study of \NPC problems lies at the heart of computer science with numerous results in algorithms and complexity. In this paper we raise the question of computing \emph{fault tolerant} solutions to \NPC problems; that is computing a solution that can survive the ``failure'' of a few constituent elements. This notion has appeared in a variety of theoretical and practical settings such as estimating network reliability, kernelization (aka instance compression), approximation algorithms and so on. In this paper, we seek to highlight these questions for further research.
        
        As a concrete example, we study the fault-tolerant version of the classical {\sc Feedback Vertex Set (FVS)} problem, that we call {\sc Fault Tolerant Feedback Vertex Set (FT-FVS)}. Recall that, in FVS the input is a graph $G$ and the objective is to compute a minimum subset of vertices $S$ such that $G-S$ is a forest. In FT-FVS, the objective is to compute a minimum weight subset $S$ of vertices such that $G - (S \setminus \{v\})$ is a forest for any $v \in V(G)$. Here the vertex $v$ denotes a \emph{single vertex fault}. We show that this problem is \NPC, and then present a constant factor approximation algorithm as well as an FPT-algorithm parameterized by the solution size.  
        We believe that the question of computing fault tolerant solutions to various \NPC problems is an interesting direction for future research.
    \end{abstract}
\end{titlepage}



\section{Introduction} \label{sec:intro}
The class of \NPC problems lies at the heart of computer science and encompasses numerous computational problems arising in various domains. Computing optimal solutions to these problems is of enormous importance in both practical and theoretical settings.
However, as this unlikely to be computationally tractable, research has focused on studying these problems in other algorithmic paradigms such as Approximation algorithms and Parameterized Complexity. In approximation algorithms, the objective is to design polynomial time algorithms that output a solution whose cost is a bounded multiplicative factor (called the approximation factor) away from an optimal solution. In Parameterized Complexity, the objective is to compute  optimal solutions to instances with a bounded structural parameter. These algorithms are known as \emph{Fixed Parameter Tractable (FPT)} algorithms.
A common parameter is an upperbound on the solution size, denoted by $k$, and an FPT-algorithm for the problem decides if the input instance admits a solution of cardinality $k$ in time $f(k) n^{\O(1)}$, where $n$ denotes the input size and $f$ is a function of $k$ alone. There is vast body of work on these topics, too many to exhaustively enumerate here, and we refer to some well-known textbooks for a broader overview \cite{vazirani2013approximation,williamson2011design,downey2013fundamentals,cygan2015parameterized}.
 
Fault-tolerant data structures are another prominent topic of research in computer science, especially for network design problems. Here we have a graph $G$, whose vertices and edges represent in network components that are not completely reliable. We are interested in the properties of the network after a subset $F$ of edges or vertices have failed, and our objective is to design efficient data structures that can answer queries in $G-F$. As an example, consider the question of \emph{$k$-fault tolerant reachability}. Here, given a graph $G$ and an integer $k$ the objective is to compute a minimum size spanning subgraph $H$ of $G$ with the following property: for any subset $F$ of upto $k$ faulty edges and any two vertices $u,v \in V(G)$, $u$ and $v$ are connected in $G-F$ if and only if they are connected in $H - F$. A well known result of Nagamochi and Ibaraki~\cite{NI92} gives a subgraph $H$ with at most $(k+1) \cdot |V(G)|$ edges, and this is optimal. Similar questions on reachability, distances and other graph properties are very well studied; we refer to some recent works \cite{patrascu2007planning,duan2017connectivity,baswana2018approximate,baswana2018fault,chakraborty2018sparse,parter2013sparse,parter2015dual,bilo2014fault,chechik2011fault,chechik2010fault,chechik2012f,bilo2015improved,dinitz2011fault,parter2017vertex,parter2016fault}.

In this paper we raise question of computing a \emph{fault tolerant solution} to an NP-Complete problem. To illustrate this notion concretely, we consider the {\sc Feedback Vertex Set} problem. Here we are given a graph $G$ and the objective is to compute a minimum \emph{feedback vertex set} of $G$: a subset of vertices $S$ such that $G - S$ is acyclic. This is a classical \NPC problem that arises in deadlock recovery, bayesian inference, VLSI design and many other applications. It was among the first problems that were shown to be \NPC. It admits a factor-2 approximation algorithm~\cite{bafna19992,chudak1998primal}, as well as an FPT algorithm running in time $2.7^k n^{\O(1)}$~\cite{li2020detecting}.
Now consider a scenario where given a graph $G$ and a feedback vertex set $S$, some subset of vertices $F \subseteq S$ fail, e.g. we are unable to delete them from $G$ for some external reasons. Then $S-F$ may no longer be a feedback vertex set of $G$. For example, in deadlock recovery we seek to resolve resource contention among a set of processes by terminating a subset of processes. The minimum subset of processes to be terminated corresponds to a feedback vertex set $S$ in an associated graph $G$. Due to external factors, we may not be able to ensure that every process in $S$ can be terminated, and hence we fail to resolve the deadlock. In such a scenario, we require a feedback vertex set that can survive the failure  of a bounded number of vertices. Formally, we have the following problem.

\probdefn{$d$-Fault Tolerant Feedback Vertex Set ($d$-FT-FVS)}{Graph $G$, weight function $w:V(G) \rightarrow \mathbb{R}^+$}{A subset $S \subseteq V(G)$ minimizing $w(S)$ such that for any subset $F \subseteq V(G)$ of up to $d$ vertices, $G - (S \setminus F)$ is a forest.}

Observe that when $d=0$, this is simply FVS. In this paper, we study {\sc $d$-Fault Tolerant Feedback Vertex Set} for single vertex faults, i.e. $d=1$. We remark that the study of single edge or vertex faults are an important class of questions in fault tolerant data structures~\cite{parter2013sparse,demetrescu2008oracles,baswana2018approximate,parter2014fault}. It also illustrates a key difficulty in the computation of fault tolerant solutions: given a graph $G$, a vertex $v \in V(G)$, let $H = G - \{v\}$ and let $S$ be a fault tolerant solution to the graph $H$; then observe that $S \cup \{v\}$ need not be a fault tolerant solution to the graph $G$. In the normal setting, it is implicitly assumed that $S \cup \{v\}$ will be a solution to $G$. Hence the known algorithms don't easily extend to the fault tolerant setting even for a single fault. 

We begin by showing that {\sc Fault Tolerant Feedback Vertex Set} is \NPC. We then present a constant-factor approximation algorithm, and an FPT-algorithm parameterized by the cardinality of the solution. More formally, our results are as follows.
\begin{theorem}
    \begin{itemize}
        \item[(i)] FT-FVS is \NPC.
        \item[(ii)] FT-FVS admits a factor-34 approximation in polynomial time. It can be improved to a factor-3 approximation in unweighted instances.
        \item[(iii)] FT-FVS parameterized by the solution size $k$ admits an FPT-algorithm that runs in time $k^{\O(k)}n^{\O(1)}$.
    \end{itemize}
\end{theorem}
 
\paragraph{Related Work.} Consider the classical {\sc MinCut} problem where given a graph $G$ and two vertices $s,t$, the goal is to compute a minimum subset of vertices $X$ such that $s$ and $t$ are disconnected in $G - X$. It is well known that this problem is solvable in polynomial time. In the \emph{fault tolerant} version of this question, the objective is to compute a minimum subset $X$ of vertices such that for any subset $F$ of upto $d$ vertices, $s$ and $t$ remain disconnected in $G - (X \setminus F)$. This question has been studied in the context of estimating network reliability, and it can be also solved in polynomial time~\cite{wagner1990disjoint,dean2011approximation}. The {\sc Multiway Cut} problem is a generalization of {\sc Mincut}, where we have a subset of vertices $T \subseteq V(G)$, called \emph{terminals}, and the goal is to remove a minimum subset of vertices $X$ such that every pair of vertices in $T$ is disconnected in $G - X$. {\sc Multiway Cut} is NP-complete even for $|T| = 3$, but it is known to admit a factor-$1.2965$ approximation algorithm~\cite{sharma2014multiway} and an FPT algorithm running in time $2^k n^{\O(1)}$~\cite{cygan2013multiway}. In {\sc $d$-Fault Tolerant Multiway Cut} the objective is to compute a minimum subset of vertices $X$ such that for any subset $F$ of upto $d$ vertices, every pair of terminal vertices are disconnected in $G - (X \setminus F)$. This question was studied in ~\cite{dean2011approximation} under a different name,  {\sc $d$-Hurdle Multiway Cut}, motivated by applications in network reliability among a set of terminal vertices. They presented a $2$-approximation algorithm for it that is optimal under UGC. Yet another class of problems where the notion of fault-tolerance appears are the {\sc Hitting Set} / {\sc Set Cover} problems. Here the input is a family of subsets $\calF$ of a universe $U$, and the goal is to compute a minimum subset of the universe $S \subseteq U$ such that $S$ has a non-empty intersection with each subset from $\calF$. In the fault tolerant setting, we require that $S$ intersects each subset from $\calF$ in at least $d$ elements. This is known as the {\sc Set Multicover Problem}~\cite{hua2009exact}. More broadly, these questions can also be captured by \emph{General Covering Integer Programs}, where the requirement of fault-tolerance can be encoded as a constraint in a ILP. This setting can capture a number of computational problems, and tight approximation results are known~\cite{kolliopoulos2001tight}.

Fault tolerant solutions to computational problems also show up in other settings. 
Kernelization is a sub-area of Parameterized Complexity that studies \emph{compression} of problem instances in polynomial time. A key technique in here is to first compute an approximate solution $S$ (called a \emph{modulator}) to the problem in the input graph $G$, and then design reduction rules that exploit structural properties of the graph $G-S$. 
There are some recent results in kernelization that introduce the notion of a \emph{$d$-redundant modulator}~\cite{AgrawalLMSZ19,AgrawalM0Z19,JansenP18}. This notion is almost equivalent to a $d$-fault tolerant solution. Such modulators allow for the design of more powerful reduction rules, since stronger structural properties of the input graph $G$ can be derived from them. 
Yet another setting where the notion of fault-tolerant solutions has appeared is an approximation algorithm for approximating the Weighted FlowTime on a Single Machine~\cite{Batra0K18}. There the problem is reduced to an instance of {\sc $d$-Fault Tolerant Multicut} in trees, named {\sc Demand Multicut} in \cite{Batra0K18}, towards obtaining an approximation algorithm.

Given the wide variety of settings where the notion of fault tolerant solutions to computational problems have appeared, we believe they constitute an interesting class of problems for further study. In this paper we introduce and study the fault tolerant version of FVS for single vertex faults, called FT-FVS. The rest of the paper is organized as follows: Section 2 presents some preliminaries and proves that FT-FVS is \NPC. Section 3 presents a constant-factor approximation algorithm, and Section 4 presents an FPT-algorithm for the problem. We conclude in Section 5 with further remarks and open problems.


\section{Preliminaries and NP-Completeness}
In this section we review a few graph theory preliminaries, and then prove the FT-FVS is \NPC. We refer to ~\cite{diestel2005} for more details.
Let $G$ be a graph
and let $V(G)$ and $E(G)$ denote the vertex set, and the edge set of $G$. The graph $G$ may have parallel edges, however it doesn't contain any self-loops.\footnote{A self-loop defines a cycle with just one vertex, and hence there can be no fault tolerant fvs in the graph} 
A graph without any parallel edges is called a \emph{simple graph}.
For a vertex subset $S \subseteq V(G)$, let $G - S$ denote the induced subgraph $G[V(G) \setminus S]$. 
The \emph{contraction} of an edge $(u,v) \in E(G)$ results in a graph, denoted by $G/(u,v)$, that is obtained by deleting $u,v$ from $G$, and adding a new vertex $w$ along with the edges $\{(w,x) \mid (u,x) \in E(G), x \neq v\} \cup \{(w,y) \mid (v,y) \in E(G), y \neq u\}$. Note that, we obtain parallel edges for any common neighbors of $u$ and $v$, and further any parallel edges between one of $u$ or $v$ and any other vertex in $V(G) \setminus \{u,v\}$ are also preserved in $G/(u,v)$.

A \emph{feedback vertex set (fvs)} of $G$ is a subset of vertices $X$ such that $G - X$ is a forest. In the {\sc Feedback Vertex Problem (FVS)} problem, the input is a graph $G$ and a weight function $w:V(G) \rightarrow \mathbb{R}^+$, and the objective is to find an fvs of minimum total weight.
When the weight function is absent, or assigns the same weight to every vertex, we say that the instance is \emph{unweighted}, otherwise it is \emph{weighed}. Next, we extend this definition to the fault tolerant setting.

\begin{definition}
    Let $G$ be a graph. A subset of vertices $S$ is called a \emph{$d$-Fault Tolerant Feedback Vertex Set ($d$-fvs)} of $G$ if for every subset of vertices $F$ of cardinality at most $d$, the graph $G - (S \setminus F)$ is a forest.    
\end{definition}

Observe that for $d=0$, $d$-fvs is just fvs.
The following observation follows immediately from the definition of $d$-fvs.
\begin{observation}
    A vertex subset $X$ is a $d$-fvs of $G$, if any only if for every cycle $C$ we have $|V(C) \cap X| \geq d+1$.
\end{observation}

In the {\sc $d$-Fault Tolerant Feedback Vertex Set problem}, the input is a graph $G$ with weight function $w:V(G) \rightarrow \mathbb{R}^+$, and the objective is to compute a $d$-fvs of $G$ of minimum total weight.
In this paper, we focus on the case of $d=1$, i.e. computing a minimum weight 1-fvs, which we call the {\sc Fault Tolerant Feedback Vertex Set (FT-FVS)} problem. 
Let us prove that this problem is \NPC.
\begin{theorem}\label{thm:FT-FVS-NPC}
    {\sc Fault Tolerant Feedback Vertex Set} is \NPC.
\end{theorem}
\begin{proof}
    It is easy to see that this problem is in \NP, and we only need to establish that it is \NPH.
    We give a reduction from the \NPC{} {\sc Vertex Cover (VC)} problem, where given a graph $G$, a subset of vertices $S$ is called a \emph{vertex cover} of $G$ if for every edge $(u,v) \in E(G)$ at least one of $u,v$ lies in $S$.
    In the {\sc Vertex Cover} problem the input is a graph $G$ and an integer $k$ the objective is to decide if there is a vertex cover of cardinality at most $k$. Given an instance $(G,k)$ of VC we construct an instance of FT-FVS as follows. We add a new vertex $r$ that is adjacent to every vertex $v \in V(G)$, and then add two new vertices $s,t$ along with edges $(r,s), (s,t), (t,r)$. Let $H$ denote the new graph, and we claim that $H$ admits a 1-fvs of cardinality at most $k+2$ if and only if $G$ admits a vertex cover of cardinality at most $k$.
    
    In the forward direction, suppose that $X$ is a 1-fvs of $G$ of cardinality at most $k+2$, and assume that it is minimal. Then $X$ contains $r$ and one of $s,t$. Let $Y = X \setminus \{r,s,t\}$ and observe that $|Y| \leq k$. Let us argue that $Y$ is a vertex cover of $G$. Suppose not, and let $(u,v) \in E(G)$ be an edge whose endpoints are disjoint from $Y$. In $H$, we have a cycle formed by $\{r,u,v\}$ and observe that $X$ contains just one vertex of this cycle. This is a contradiction. Hence $Y$ must be a vertex cover of $G$ of cardinality $k$.
    
    In the reverse direction, given a vertex cover $Y$ of $G$ of cardinality at most $k$, we claim that $X = Y \cup \{r, s\}$ is a 1-fvs of $H$ (clearly of cardinality at most $k+2$). Consider any cycle $C$ of $H$ that doesn't contain the vertex $r$. Then $C$ is a cycle contained in the subgraph $G$ of $H$, and it is easy to see that the vertex cover $Y$ contains at least two vertices from $V(C)$. If the cycle $C$ contains $r$, then either it is the cycle ${r,s,t}$ or else it contains at least one edge $(u,v) \in E(G)$. Therefore $X$ contains at least two vertices of $C$. Hence, $X$ is a 1-fvs of $G$ of cardinality $k+2$.
\end{proof}

\section{Approximation Algorithm}
In this section we present a constant factor approximation algorithm for FT-FVS. Let us recall that for a minimization problem, a factor-$\alpha$ approximation algorithm runs in polynomial time and produces a solution $S$ such that $cost(S) \leq \alpha \cdot cost(OPT)$, where $OPT$ is an optimal solution to the problem instance. Our algorithms for the unweighted and the weighted versions of the problem differ slightly, and we obtain a factor-3 and a factor-34 approximation, respectively. Let us remark that FVS admits a factor-2 approximation that is optimal (under UGC)\cite{bafna19992,chudak1998primal}. Observe that a lower-bound of 2 on the approximation factor for FT-FVS follows by our reduction in the proof of  Theorem~\ref{thm:FT-FVS-NPC}, and a lower-bound on the approximation factor of {\sc Vertex Cover} based on the \emph{Unique Games Conjecture}.
\begin{theorem}
    {\sc FT-FVS} cannot be approximated within factor-$(2 - \delta)$, for any constant $\delta > 0$ by a polynomial time algorithm, under the Unique Games Conjecture.
\end{theorem}

The main result of this section is formally stated as follows.

\begin{theorem}\label{thm:FT-FVS-Approx}
    Let $G$ be a graph and let $w:V(G) \rightarrow \mathbb{R}^+$ be a weight function on the vertices. There is a polynomial time algorithm that computes a 1-fvs $S$ of $G$ such that $w(S) \leq 34 w(S^\star)$ where $S^\star$ is a minimum weight 1-fvs of $G$. 
    In unweighted instances we obtain a factor-3 approximation.
\end{theorem}
In the remainder of this section we present a proof of this theorem.
    Let $S$ denote the approximate 1-fvs that we compute.
    Let $\mu^\star$ be the weight of a minimum weight 1-fvs of $G$. In unweighted instances, $\mu^*$ denotes the cardinality of a minimum 1-fvs of $G$.
    The first step of our algorithm is to compute an approximate fvs of $G$ and add it to $S$. We require the following proposition.
    \begin{proposition}[\cite{bafna19992,chudak1998primal}]
        {\sc Weighted Feedback Vertex Set} has a factor-2 approximation algorithm that runs in polynomial time.
    \end{proposition}
    The following claim is immediate from the fact that the weight of a minimum fvs of $G$ is upper-bounded by $\mu^\star$.
    \begin{observation}
        Let $X$ be a fvs of $G$ such that $w(X)  \leq \alpha w(X^\star)$,
        where $X^\star$ is a minimum weight fvs of $G$. Then $w(X)  \leq \alpha  \mu^\star$.
    \end{observation}
    
    Let $X$ denote a 2-approximate fvs of $(G,w)$ that we add to $S$. Note that $X$ intersects every cycle of $G$ in at least one vertex. Then that any cycle $C$ of $G$ such that $|V(C) \cap X| \geq 2$ is already covered by $X \subseteq S$. It only remains cover those cycles of $G$ that intersect $X$ in at most one vertex. The following claim is immediate from the fact that $G-X$ is a forest.
    \begin{observation}
        Let $v \in X$ and let $C$ be a cycle in $G$ such that $V(C) \cap X = \{v\}$. Then there is a tree $T$ in the forest $G - F$ such that $C$ is contained in the subgraph induced by $V(T) \cup \{v\}$.
    \end{observation}
    Let $\calC(X)$ denote the collection of cycles in $G$ that contain exactly one vertex of $X$. Our next step is to compute a subset of vertices $Y$, disjoint from $X$, that intersects every cycle in $\calC(X)$. We will accomplish this task by reducing it the the {\sc Vertex Multicut} problem. Here, the input is a $H$ with weights $w_H$ on the vertices, a collection $\{s_i, t_i\}$ of \emph{terminal pairs} of vertices. A subset of vertices $Z$ is called a multicut if every terminal pair $s_i, t_i$ is disconnected in $G-Z$. In the {\sc Vertex Multicut} problem the objective is to compute a minimum weight multicut. This problem is \NPC and admits a factor-$\O(\log n)$ approximation algorithm in general graphs~\cite{vazirani2013approximation,williamson2011design}.
    We reduce the problem of hitting every cycle in $\calC(X)$ to {\sc Vertex~Multicut}~in~\emph{forests} as follows.
    Let us consider $H = G - X$ with weight function $w$~(restricted to V(H)), and consider the following collection of vertex pairs.
        $$\calP(X) = \{ \{s,t\} \subseteq V(T) \cap N(v) \mid v \in X \text{ and } T \text{ is a tree in } G - X \}$$
    \begin{lemma}\label{lemma:approx}
        Let $S$ be any 1-fvs of $G$, then $V(H) \cap S$ is a multicut for $H$ and $\calP(X)$. Conversely, if $Y$ is any multicut for $H$ and $\calP(X)$, then $X \cup Y$ is a 1-fvs of $G$.
    \end{lemma}
    \begin{proof}
        In the forward direction, consider any pair of vertices $\{s,t\} \in \calP(X)$ that lie in $V(T) \cap N(v)$ for some vertex $v \in X$ and some tree $T$ in $G-X$.
        Then there is a cycle $C$ formed by $v,s,t$ and the path between $s$ and $t$ in $T$, that contains exactly one vertex of $X$. Observe that, as $|S \cap V(C)| \geq 2$, $S$ contains at least one vertex from the path in $T$ between $s$ and $t$. Since this is true for every pair in $\calP(X)$, $V(H) \cap S$ is a multicut.
        
        In the reverse direction, consider any cycle $C$ of $G$, and let us argue that $X \cup Y$ contains at least 2 vertices of $C$.
        If $|V(C) \cap X| \geq 2$, this is clearly true. Otherwise, as $X$ is an fvs of $G$, $V(C) \cap X = \{v\}$. Then, there is a tree $T$ in the forest $G-X$ such that $C$ is contained in the subgraph induced by $V(T) \cup \{v\}$. Let $s,t$ be the neighbors of $v$ in the cycle $C$. Observe that $s,t \in V(T) \cap N(v)$, and hence $\{s,t\} \in \calP(X)$. Further note that the cycle $C$ consists of $v,s,t$ and the path $Q$ in $T$ between $s$ and $t$.
        Since, $Y$ is a multicut for $H$ and $\calP(X)$, there is a vertex $u \in V(Q) \cap Y$. Hence $\{u,v\} \subseteq (X \cup Y) \cap V(C)$.
        Therefore we conclude that $X \cup Y$ is a 1-fvs of $G$.
    \end{proof}
    
    Let $Y^\star$ denote a minimum-weight multicut for $H,w, \calP(X)$. It is clear that $w(Y^\star) \leq \mu^\star$. It only remains to compute one such multicut for $H,w,\calP(X)$. The unweighted instances of {\sc Vertex Multicut} on forests can be solved in polynomial time via a simple greedy algorithm, that roots each tree of the forest at an arbitrary vertex and then iteratively picks the lowest least common ancestor of a terminal pair and deletes it from the graph.
    \begin{proposition}[see e.g.~\cite{guo2006complexity}]
        Unweighted {\sc Vertex Multicut} on forests can be solved optimally in polynomial time.
    \end{proposition}
    \noindent
    Therefore, in unweighted instances, let $Y$ denote an optimal multicut for $H$ and $\calP(X)$ that can be computed in polynomial time.
    However, weighted {\sc Vertex Multicut} in forests in \NPC. We can easily obtain a reduction from {\sc Vertex Cover}.\footnote{Given an instance $G$ of {\sc Vertex Cover}, construct a \emph{star graph} $H$ with $|V(G)|$ leaves, each associated with a vertex of $G$. Let $r$ denote the center vertex of the star. Each edge $(u,v) \in E(G)$ defines a terminal pair of the {\sc Vertex Multicut} instance, and let $\calP$ denote this collection. Finally, define a weight function $w:V(H) \rightarrow \mathbb{R}^+$, that gives weight $1$ to each leaf and weight $n+1$ to the center vertex $r$. It is easy to see that any vertex cover $S$ of $G$ corresponds to a {\sc Vertex Multicut} of $H$ and $\calP$ of weight $w(S) \leq n$, formed by the leaves of $H$ corresponding to $S$. Conversely, any multicut $S$ of weight at most $n$ of $H$ and $\calP$, gives a vertex cover of cardinality $w(S)$ of $G$.}
    Therefore, we require the following proposition that gives an approximation algorithm for Weighted {\sc Vertex Multicut} in forests.
    \begin{proposition}[\cite{agrawal2020polylogarithmic}]
        Weighted {\sc Vertex Multicut} on forests admits a factor-32 approximation algorithm in polynomial time.
    \end{proposition}
    We remark that the above approximation algorithm is actually for the larger class of Chordal graphs. Let $Y$ denote the multicut to $(H,w,\calP(X))$ output by the above approximation algorithm. It is clear that $w(Y) \leq 32\mu^\star$.
    
    Having computed $X$ and $Y$ we output $S = X \cup Y$ as the required solution. It follows from Lemma~\ref{lemma:approx} that $X \cup Y$ is a $1$-fvs of $G$, and it follows from the above discussion that $w(X \cup Y) \leq 34 \mu^\star$ in weighted instances of FT-FVS. When the instance is unweighted, $|X \cup Y| \leq 3 \mu^\star$. This concludes the proof of Theorem~\ref{thm:FT-FVS-Approx}.  

\section{FPT Algorithm}
In this section we present an FPT algorithm for FT-FVS parameterized by the solution size, that is denoted by $k$. Our input is a graph $G$ on $n$ vertices and an integer $k$, and our objective is to compute a 1-fvs of $G$ of cardinality at most $k$, if one exists, in time $f(k) n^{\O(1)}$. Here $f$ is some function of $k$ alone. To present our algorithm, it is helpful to consider a slightly more general problem defined below.

\paradefn{Fault Tolerant FVS Extension (FT-FVS-Ext)}{Graph $G$, an integer $k$ and a vertex subset $S \subseteq V(G)$.}
{k}{Is there a 1-fvs $X$ of $G$ such that $|X| \leq k$ and $S \subseteq X$?}

\noindent
Note that, when $S=\emptyset$, {\sc FT-FVS-Ext} is just FT-FVS.
Let us introduce a few terms to simplify the presentation of our algorithm. We call a vertex subset $X \subseteq V(G)$ a \emph{solution} to $(G,S,k)$ of {\sc FT-FVS-Ext}, if $X$ is a 1-fvs of cardinality at most $k$ that is also a superset of $S$. We call $(G,S,k)$ a \emph{\No-instance} if there is no solution, i.e. there is no 1-fvs of cardinality $k$ that is a superset of $S$; otherwise it is a \emph{\Yes-instance}. 
We present a branching algorithm for this problem, where we use the vertex subset $S$ to keep track of the current partial solution. 
We remark that, in the normal setting, i.e. for FVS, we could have simply deleted $S$ from the graph and then recursively solved the problem on the remaining subgraph. However, this would not work for FT-FVS since, upon deleting the  $S$, we also lose all those cycles that contain exactly one vertex of the current partial solution $S$. Hence, we must retain all vertices of $S$ until we a complete solution is obtained.

Let $(G,S,k)$ denote an instance of {\sc FT-FVS-Ext}.
We begin by presenting a set of reduction rules that simply this instance and make it amenable to branching. 

\begin{reduction}\label{rr:trivial}
    Let $(G,S,k)$ be an instance of {\sc FT-FVS-Ext}.
    If $|S| > k$, then output that this is a \No-instance.
    Othereise, if $S$ is a 1-fvs of $G$ then output that this is a \Yes-instance.
\end{reduction}

\begin{reduction}\label{rr:2-cycle}
    Let $(G,S,k)$ be an instance of {\sc FT-FVS-Ext}.
    If there is a pair of vertices $u,v \in V(G)$ with at least 2 parallel edges between them, add $u,v$ to $S$, and remove all but one edge between them.
\end{reduction}

\begin{reduction}\label{rr:acyclic-vertex}
    Let $(G,S,k)$ be an instance of {\sc FT-FVS-Ext}.
    If $v \in V(G) \setminus S$ is not part of any cycle, then delete $v$ from $G$. 
\end{reduction}

\begin{reduction}\label{rr:deg-2}
    Let $(G,S,k)$ be an instance of {\sc FT-FVS-Ext}.
    Let $v \in V(G) \setminus S$ be a vertex of degree $2$ in $G$, with neighbors $u$ and $w$.
    Then, delete $v$ from $G$ and add an edge $(u,w)$ to $G$.  
\end{reduction}

It is clear that Reduction Rules~\ref{rr:trivial}, \ref{rr:2-cycle}, \ref{rr:acyclic-vertex} and \ref{rr:deg-2} are safe, and can be applied in polynomial time. The next reduction rules further simplifies the graph $G$ using the partial solution $S$.

\begin{reduction}\label{rr:S-leaf}
    Let $(G,S,k)$ be an instance of {\sc FT-FVS-Ext} Reduction Rules~\ref{rr:2-cycle}, \ref{rr:acyclic-vertex} and \ref{rr:deg-2} are 
    not applicable.
    \begin{itemize}
        \item If there is a vertex $v \in V(G) \setminus S$ such that $N_G(v) \subseteq S$, then delete $v$ from $G$ and add an edge between each pair of vertices in $N_G(v)$.
        \item If there is a vertex $v \in V(G) \setminus S$ such that $N_G(v) \setminus S = \{u\}$, then contract the edge $(u,v)$.
    \end{itemize}
\end{reduction}
\begin{lemma}
    Reduction Rule~\ref{rr:S-leaf} is safe and applicable in polynomial time.
\end{lemma}
\begin{proof} 
    When Reduction Rules~\ref{rr:2-cycle}, \ref{rr:acyclic-vertex} and \ref{rr:deg-2} are not applicable, then $G$ is a simple graph where every vertex lies in some cycle.
    Consider the first case, i.e. there is a vertex $v \in V(G) \setminus S$ such that $N_G(v) \subseteq S$. Then observe that any cycle that contains $v$ also contains at least 2 vertices of $S$. Let $G'$ denote the graph obtained by deleting $v$ and adding an edge between each pair of vertices in $N_G(v) \subseteq S$. It is clear that any solution to the instance $(G',S,k)$ is a solution to $(G,S,k)$, and vice-versa.
    
    Next, consider the second case, i.e. $N_G(v) \setminus S = \{u\}$. Note that Reduction Rule~\ref{rr:2-cycle} is not applicable to $(G,S,k)$, and therefore there are no parallel edges between $u$ any of it's neighbors in $G$. Hence any cycle $C$ containing $v$ must either contain at least 2 neighbors of $v$ from $S$, or contain $u$ and a neighbor from $S$. The first type of cycles are always covered by any solution. And the second type of cycles must contain the edge $(u,v)$. Therefore, if $X$ is a solution to $(G,S,k)$ that contains $v$, then $X' = (X \setminus \{v\}) \cup \{u\}$ is also solution to $(G,S,k)$.     
    Then it is easy to see that $(G,S,k)$ has a solution if and only if $(G/(u,v),S,k)$ has a solution.
\end{proof}
We remark that an application of Reduction Rule~\ref{rr:deg-2} and ~\ref{rr:S-leaf} may create new parallel edges in $G$, which are then eliminated by Reduction Rule~\ref{rr:2-cycle}. The following lemma enumerates some useful properties of the instance when none of the above reduction rules are applicable.

\begin{lemma} \label{lemma:rr-prop}
    If Reduction Rules~\ref{rr:trivial}, \ref{rr:2-cycle}, \ref{rr:acyclic-vertex}, \ref{rr:deg-2} and ~\ref{rr:S-leaf} are not applicable to an instance $(G,S,k)$, then we have the following properties:
    \begin{itemize}
        \item $G$ is a simple graph such that every vertex of $V(G) \setminus S$ lies in some cycle, 
        \item the minimum degree in $G$ of any vertex in $V(G) \setminus S$ is at least $3$, 
        \item $|S| \leq k$ and it is not a 1-fvs of $G$.
        \item and the graph $G_S = G - S$ has minimum degree at least $2$, and it contains a cycle.
    \end{itemize}
\end{lemma}
\begin{proof}
    The first three properties are easily follow from the reduction rules.
    Observe that if $G$ has parallel edges between two vertices, then Reduction Rule~\ref{rr:2-cycle} is applicable. Similarly, if there is vertex of degree 0,1 or 2 in $V(G) \setminus S$, then Reduction Rules ~\ref{rr:acyclic-vertex} or \ref{rr:deg-2} are applicable.
    Furthermore, if some vertex in $V(G) \setminus S$ is not part of any cycle, then Reduction Rule~\ref{rr:acyclic-vertex} is applicable. And finally, if Reduction Rules~\ref{rr:trivial} is not applicable then $|S| \leq k$ and $S$ is not a 1-fvs of $G$.
    
    For the fourth property, consider the graph $G_S$. If this is an empty graph, then clearly every cycle of $G$ contains 2 vertices from $S$; i.e. Reduction Rule~\ref{rr:trivial} is applicable. Otherwise, if $v \in V(G_S)$ is a vertex of degree 1 in $G_S$, then Reduction Rule~\ref{rr:S-leaf} applies. Hence, if these reduction rules are not applicable then the minimum degree of $G_S$ is at least $2$. Then we can conclude that $G_S$ contains a cycle, otherwise $G_S$ is a forest and there must be a vertex of degree 1 in it. 
\end{proof}

\begin{lemma}\label{lemma:bdd-cycle}
    Let $(G,S,k)$ be an instance such that Reduction Rules~\ref{rr:trivial}, \ref{rr:2-cycle}, \ref{rr:acyclic-vertex}, \ref{rr:deg-2} and ~\ref{rr:S-leaf} are not applicable to it. 
    Then in polynomial time we can conclude that
        either $(G,S,k)$ is a \No-instance,
        or $S$ contains a cycle of length at most $2(2k-1)|S| \log n$ where $n = |V(G)|$.
\end{lemma}
\begin{proof}
    By Lemma~\ref{lemma:rr-prop}, the minimum degree of $G_S$ is at least 2, and let $T = \{ u \in V(G_S) \mid \deg_{G_S}(u) = 2\}$.
    Observe that any vertex $u \in T$ must have at least one neighbor in $S$ (in the graph $G$); otherwise it is a vertex of degree 2 in $G$ from $V(G) \setminus S$, and Reduction Rule~\ref{rr:deg-2} is applicable.
    We claim that if the induced subgraph $G_S[T]$ has a path of length at least $(2k - 1)|S| + 1$, then $(G,S,k)$ is a \No-instance. Towards this, consider a $P$ be a path in $G_S[T]$ of length at least $(2k-1)|S| + 1$. Since every vertex of $P$ has a neighbor in $S$ (in the graph $G$), there must be a vertex $v \in S$ with at least $2k$ neighbors in $P$. Therefore, there is a collection of $k$ cycles in $G$, each containing $v$, such that they contain no other vertex from $S$, and they are pairwise vertex disjoint except for the vertex $v$. Then it follows that any 1-fvs of $G$ must contain at least $k+1$ vertices, and hence $(G,S,k)$ is a \No-instance.
    
    Otherwise, every path in $G_S[T]$ is of length at most $(2k-1)|S|$. Further, the maximum degree of $G_S[T]$ is at most 2, and hence it union of vertex disjoint paths and cycles. If $G_S[T]$ contains a cycle, then it is of length at most $(2k-1)|S|$, as required by the lemma. Otherwise, $G_S[T]$ is a collection of vertex disjoint paths, each of length at most $2(2k-1)|S|$.
    Let $P$ be a maximal path in $G_S[T]$ and note that only the endpoints of the path have a neighbor in $G_S$, exactly one each. If there is a path $P$ in $G_S[T]$ whose both endpoints are adjacent to the same vertex in $G_S$, then we have a cycle of length $(2k-1)|S| + 1$, as required by the lemma. Otherwise, the two endpoints of every maximal path $P$ in $G_S[T]$ are adjacent to two distinct vertices. 
    
    Consider the graph $G'$ obtained from $G_S$ contracting each edge incident on a vertex in $T$. Observe that each edge of $G'$ is either an edge in $G_S$, or it corresponds to a distinct maximal path in $G_S[T]$ (that has been contracted to this edge).
    It is clear that $G'$ is a graph of minimum degree 3, and therefore $G'$ contains a cycle of length at most $2 \log |V(G')| \leq 2 \log n$ ~\cite{diestel2005}[Cor.~1.3.5]. Let $C'$ be such a cycle in $G'$, and note that it can be constructed by a BFS-traversal of $G'$. Now consider the cycle $C$ in $G_S$ obtained from it by replacing each edge of $C'$ that is not an edge in $G_S$ by the corresponding path in $G_S[Z]$. 
    The cycle $C$ has length at most $2 \log n \cdot (2k - 1)|S|$ in $G_S$, as required by the lemma.    
\end{proof}

Now we are ready to show that {\sc FT-FVS-Ext} admits an FPT algorithm.
\begin{theorem}\label{thm:fvs-fpt}
    Let $(G,S,k)$ be an instance of {\sc FT-FVS-Ext}.
    Then in time $k^{\O(k)}n^{\O(1)}$ we can compute a solution to this instance, or conclude that no such solution exists.
\end{theorem}
\begin{proof}
    Given an instance $(G,S,k)$, we apply Reduction Rules~\ref{rr:trivial}, \ref{rr:2-cycle}, \ref{rr:acyclic-vertex}, \ref{rr:deg-2} and \ref{rr:S-leaf} exhaustively. Then, either we conclude that it is a \No-instance, or it is a \Yes-instance (with solution $S$), or else we have an instance where none of these reduction rules are applicable.  Then by Lemma~\ref{lemma:bdd-cycle}, either we conclude that it is a \No-instance, or we compute a cycle $C$ in $G_S$ of length at most $2(2k-1)|S|\log n \leq 4k^2 \log n$ (as $|S| \leq k$).
    Since, $C$ is disjoint from $S$, at least 2 vertices of $C$ must be in any solution to this instance. We branch on the choice of one of these vertices. For each choice $v \in V(C)$, we recursively solve the instance $(G,S\cup{v},k)$. If any one of these recursive calls return a solution, then we output that as our solution. Otherwise we output that this is a \No-instance. The correctness of this algorithm easily follows from the Reduction Rules, Lemma~\ref{lemma:rr-prop} and \ref{lemma:bdd-cycle}.
    
    To bound the running time of this algorithm, observe that we have a branching algorithm, where at each node of the branching-tree we have at most $4k^2 \log n$ child nodes, one for each recursive call. The height of this branching-tree is upperbounded by $k - |S|$, since in each recursive call the cardinality of $S$ increases by 1. Therefore the running time is upper-bounded by $(4k^2 \log n)^{k - |S|} n^{\O(1)}$ which is upperbounded by $k^{\O(k)}n^{\O(1)}$ (see e.g.~\cite{cygan2015parameterized}).    
\end{proof}
As a corollary it follows that FT-FVS admits an FPT-algorithm in time $k^{\O(k)}n^{\O(1)}$.


\section{Conclusion}\label{sec:conclusion}
In this paper we combine the notion of fault tolerance with the study of the various computational problems in \NP, and raise the question of computing fault tolerant solutions. As a concrete example, we study the {\sc Fault Tolerant Feedback Vertex Set} and obtain the following results:
\begin{itemize}
    \item A factor-34 approximation algorithm, that improves to factor-3 in the unweighted version.
    \item An FPT-algorithm parameterized by the solution size $k$, that runs in time $k^{\O(k)}n^{\O(1)}$, where $n$ is the size of the input graph.
\end{itemize}
Two immediate follow-up questions are:
\begin{itemize}
    \item Is there a factor-2 approximation algorithm for FT-FVS? 
    \item Is there an FPT-algorithm for FT-FVS that runs in time $c^k n^{\O(1)}$, where $c$ is some fixed constant?
\end{itemize}
Next arises the more general problem of {\sc $d$-Fault Tolerant Feedback Vertex Set} for $d \geq 2$, whose complexity is unknown. We remark that the algorithms presented here for FT-FVS crucially depend on the fact that $d=1$. For example, the reduction to {\sc Vertex Multicut} in the approximation algorithm for FT-FVS does not work for $d \geq 2$. If one were to follow the same approach for $d \geq 2$, we end up with a version of {\sc Fault Tolerant Multicut}, where each pair of terminals $(s_i,t_i)$ has a fault-tolerance demand $d_i \leq d$. This version of {\sc Fault Tolerant Multicut} is a very interesting open question in its own right.\footnote{Some partial results on this question are presented in ~\cite{dean2011approximation} and ~\cite{Batra0K18}.}
 
More generally, the complexity of the fault-tolerant versions of many other problems such as {\sc Subset-FVS}, {\sc Odd Cycle Transversal}, {\sc Group Feedback Vertex Set}, {\sc Directed Feedback Vertex/Arc Set}, {\sc (Directed) Multicut} and {\sc Multiway Cut}, {\sc Bisection}, {\sc Cluster Editing} and so on are interesting open questions for future research.

\bibliographystyle{plain}
\bibliography{ref}

\end{document}